\documentclass{amsart}

\usepackage{amssymb, amsmath, amsthm}
\usepackage{wasysym}
\usepackage{multicol}
\usepackage{algorithm}
\usepackage{algorithmic}
\usepackage{tikz}

\newtheorem{theorem}{Theorem}
\newtheorem{lemma}[theorem]{Lemma}

\theoremstyle{remark}

\newtheorem{example}[theorem]{Example}

\newcommand{\avbl}[2]{\alpha_#2(#1) }

\newcommand{\floor}[1]{\left\lfloor #1 \right\rfloor}
\newcommand{\ceil}[1]{\left\lceil #1 \right\rceil}

\newcommand{\FF}{{\mathcal F}}
\newcommand{\FH}{{F}}
\newcommand{\HH}{{H}}
\newcommand{\cc}{{\mathfrak c}}
\newcommand{\GG}{{\mathcal G}}
\newcommand{\GH}{{G}}
\newcommand{\FW}{{\rm FW}}

\newcommand{\bigo}{{\mathcal O}}

\newcommand{\divides}{\mathbin{|} }

\newcommand{\pref}[1][]{{\rm pref}_{#1}}
\newcommand{\comment}[1]{}

\newcommand{\abs}[1]{\left\vert #1 \right\vert}

\begin{document}

\title{Periods and borders of random words}

\author{\v{S}t\v{e}p\'an Holub}
\address{
Department of Algebra, Faculty of Mathematics and Physics, Charles University, Prague, Czech Republic}
\email{holub@karlin.mff.cuni.cz}

\author{Jeffrey Shallit}
\address{
School of Computer Science,
University of Waterloo\\
Waterloo, ON  N2L 3G1 Canada}
\email{shallit@cs.uwaterloo.ca} 

\begin{abstract}
We investigate the behavior of the periods and border lengths of random
words over a fixed alphabet.  We show that the asymptotic probability
that a random word has a given maximal border length $k$ is a
constant, depending only on $k$ and the alphabet size $\ell$.  We give
a recurrence that allows us to determine these constants with any
required precision.  This also allows us to evaluate the expected
period of a random word.  For the binary case, the expected period is
asymptotically about $n-1.641$. We also give explicit formulas for the
probability that a random word is unbordered or has maximum border
length one.
\end{abstract}

\maketitle

\section{Introduction and notation}
A {\it word} is a finite sequence of letter chosen from a finite
alphabet $\Sigma$.
The periodicity of words is a classical and well-studied topic in both
discrete mathematics and combinatorics on words, starting with the classic
paper of Fine and Wilf \cite{Fine&Wilf:1965} and continuing with the
works of Guibas and Olydzko \cite{Guibas&Odlyzko:1981a,Guibas&Odlyzko:1981b,Guibas:1985}.  For more
recent work, see, for example, \cite{Halava&Harju&Ilie:2000,Tijdeman&Zamboni:2003,Rivals&Rahmann:2003}.

We say that a word $w$ has period $p$ if $w[i] = w[i+p]$ for all $i$
that make the equation meaningful.  (If $|w| = n$ and one indexes
beginning at position $1$, this would be for $1 \leq i \leq n-p$.)  
Trivially every word of length $n$ has all periods of length $\geq n$,
so we restrict our attention to periods $\leq n$.
The least period is sometimes called {\it the} period.  
For example,
the French word {\tt entente} has periods
$3$, $6$, and $7$.

Empirically, one quickly discovers that a randomly chosen word typically
has a least period that is very close to its length.
This readily follows from the fact
that the number of words over a given alphabet grows exponentially as the
length increases.  It can also be
seen as a particular case of the fact that most strings are
not compressible.

In this paper, we quantify this basic observation
and show that the expected least period of a string of length $n$ over an
$\ell$-letter alphabet  is $n-\avbl n \ell$, where $\avbl n \ell$ is
$\bigo(1)$.

Another concept frequently studied in formal language theory is
that of {\it border} of a word 
\cite{Silberger:1971,Silberger:1980}.
A word $x$ has border $w$ if $w$ is
both a prefix and a suffix of $x$.  Normally we do not consider
the trivial borders of length $0$ or $n = |w|$.  Thus, for example, the 
English word {\tt ionization} has one border:   {\tt ion}.  Less trivially,
the word {\tt alfalfa} has two borders:  {\tt a} and {\tt alfa}.  
A word with no borders is {\it unbordered}.

There is an obvious connection between periods of a word and its
borders: if $w$ has a period $p$, then it has a border of length
$|w|-p$. For example, 
the English word  {\tt abracadabra}, of length $11$, has periods
$7$, $10$, and $11$, while it has borders of length $1$ and $4$.

Consequently, the least period of a word corresponds to the length
of the longest border (and an unbordered word corresponds to least period
$n$, the length of the word).
The reader should be constantly aware of this
duality, since it is often useful and more natural to
think about periods in terms of borders. This can be seen from the
announced result: it is more compact to speak directly about the
expected maximum border length, which is $\avbl n \ell$.

If $P$ is a set of integers, we shall write $n-P$ for $\{n-p\mid p\in
P\}$, and $P-n$ for $\{p-n\mid p\in P\}$.

By $\pref[i](v)$, we mean the prefix of length $i$ of the word $v$.

\section{Multiperiodic words and the average border length}

We shall obtain our results by counting words with a given length $n$ and a given finite set of periods $P\subseteq\{1,2,\dots n\}$, or equivalently, with a given set of border lengths $n-P$. For technical reasons, in order to be able to deal with unbordered words, we shall always suppose that $n\in P$, that is, we shall say that every word has a border length zero. 

There are two basic types of requirements. Let 
\[\GG_\ell(P,n)=\{w\in \Sigma_\ell^n \mid \text{for each $p$ in $P$, $p$ is a period of $w$} \}\,, \]
and let $\GH_\ell(P,n)$ be the cardinality of $\GG_\ell(P,n)$.
Similarly, let
\[\FF_\ell(P,n)=\{w\in \GG_\ell(P,n) \mid \text{$\min P$ is the least period of $w$} \}\,, \]
and let $\FH(P,n)$  be the cardinality of $\FF_\ell(P,n)$.

Words with many periods have been amply studied. In particular, there is a fast algorithm constructing a word of length $n$ with periods $P$ and maximal possible number of letters. Such a word, called an $\FW$-word in the literature, is unique up to renaming of the letters. Let $\cc(P,n)$ denote the cardinality of the alphabet of the $\FW$-word of length $n$ and periods $P$.

\begin{example}
Let $P=\{p,q\}$ and $d=\gcd(p,q)$. The well-known periodicity lemma 
(often called the Fine and Wilf theorem, which is the origin of  the term $\FW$-word) states that if a word of length at least $p+q-d$ has periods $p$ and $q$,
then it also has period $d$.
Moreover, the bound $p+q-d$ is sharp; for all $p, q \geq 1$
there are words of length $p+q-d-1$ with period $p$ and $q$ but not period
$d$.  
This can be stated, using the just-introduced terminology,
by the two assertions $\cc(\{p,q\},p+q-d)=d$ and $\cc(\{p,q\},p+q-d-1)>d$.
\end{example}
   
The number $\cc(P,n)$ can be computed and the corresponding $\FW$-word
constructed using the algorithm of Tijdeman and Zamboni \cite{TZ} (see
\cite{multi} for an alternative presentation). The computation is
summarized by the following formula:

\begin{align*}
\cc(P,n)=
\begin{cases}
1, & \text{if $m=1$};\\
n, & \text{if $m\geq n$};\\
\cc(Q,n-m), & \text{if $2m\leq n$};\\
\cc(Q,n-m) +2m -n, & \text{if $m< n < 2m$};
\end{cases}
\end{align*}
where $m=\min P$ and $Q=(P-m)\setminus\{0\}\cup\{m\}$. 

Since each word having the periods in $P$ (and possibly others)
results from a coding (a letter-to-letter mapping)
of the corresponding $\FW$-word, we obtain
\[\GH_\ell(P,n)=\ell^{\cc(P,n)}\,,\]
 which is the starting point of our computation.

Note that $\FF_\ell(\{p\},n)$ is the set of words from $\Sigma_\ell^n$ 
having least period $p$.
Equivalently,  $\FF_\ell(\{n-r\},n)$ is the set of words with the longest border of length $r$. For $0\leq r<n$, let
\[\lambda_\ell(r,n)=\frac {\FH(\{n-r\},n)}{\ell^n}\]
denote the relative number of such words.
Our goal is to compute
\[\avbl n \ell=\sum_{r=0}^{n-1} r\cdot \lambda_\ell(r,n), \]
which is the expected maximum border length for words in $\Sigma_\ell^n$.
We first show that this quantity converges as $n$ approaches infinity.

\begin{lemma}\label{limit}
For each $\ell\geq 2$ and each $0\leq r<n$,
\begin{align*}
	&\abs{\lambda_\ell(r,n+1)-\lambda_\ell(r,n)} \leq \frac 1{\ell^{\floor{n/2}}}\,.
\end{align*}
\end{lemma}
\begin{proof}

Case 1:  $r \geq \lfloor n/2 \rfloor$.  Then
\begin{align*}
	\abs{\lambda_\ell(r,n+1)-\lambda_\ell(r,n)}&=\abs{\frac{\FH_\ell(\{n+1-r\},n+1)}{\ell^{n+1}}-\frac{\FH_\ell(\{n-r\},n)}{\ell^{n}}}\\
	&=\frac 1 {\ell^{n+1}}\abs{\FH_\ell(\{n+1-r\},n+1)-\ell\cdot \FH_\ell(\{n-r\},n)}.
\end{align*}
Recall that $\FH_\ell(\{n+1-r\},n+1)$ (resp., $\FH_\ell(\{n-r\},n)$)
counts the words with longest border length $r$ from $\Sigma_\ell^{n+1}$ (resp.,
$\Sigma_\ell^n$).
First, note that $F_\ell(\{p\},n)\leq \ell^p$ for any $p$ and $n$. This implies \[\abs{\lambda_\ell(r,n+1)-\lambda_\ell(r,n)}\leq \frac 1{\ell^{r}}\,\]
and we are done.

\medskip

Case 2:  $r<\floor{n/2}$.
There is a useful correspondence
between $\Sigma_\ell^n$ and $\Sigma_\ell^{n+1}$,
given by the insertion of a letter in the middle of the shorter word.
The basic observation, already used in \cite{harborth,Nielsen:1973},
is that this insertion does not influence borders
of length at most $\floor{n/2}$.
Define 
\begin{align*}
\FF &=\FF_\ell(\{n+1-r\},n+1),\\
\+B &=\{w_1aw_2 \mid a\in \Sigma_\ell,\, |w_1|=\floor{n/2},\, |w_2|=\ceil{n/2},\, w_1w_2\in \FF_\ell(\{n-r\},n)
\}\,.
\end{align*}
Then $\abs{\+B}=\ell\cdot \FH_\ell(\{n-r\},n)$. Let $w\in \+F\setminus \+B$ 
and write $w=w_1aw_2$ with $a\in \Sigma_\ell$, $|w_1|=\floor{n/2}$, and
$|w_2|=\ceil{n/2}$. The words $w$ and $w_1w_2$ have the same borders up
to length $\floor{n/2}$. Since $w_1w_2\notin \FF_\ell(\{n-r\},n)$, we
deduce that $w_1w_2$ has a border of length at least $\floor{n/2}+1$,
that is, a period at most $\ceil{n/2}-1$. This implies
\begin{align}
\abs{\+F\setminus \+B}\leq \ell\cdot \sum_{j=0}^{\ceil{n/2}-1}\ell^j< \ell^{\ceil{n/2}+1}.
\end{align} 
Similarly, a word  $w\in \+B\setminus\+F$ has period at most $\ceil{n/2}$, and
so
\begin{align}
\abs{\+B\setminus \+F}\leq \sum_{j=0}^{\ceil{n/2}}\ell^j< \ell^{\ceil{n/2}+1}.
\end{align} 
 We thus obtain
\begin{align*}
	&\abs{\lambda_\ell(r,n+1)-\lambda_\ell(r,n)}=\frac 1{\ell^{n+1}} \Big|\abs{\+B\setminus \+F}-\abs{\+F\setminus \+B}\Big|
	<\frac 1 {\ell^{\floor{n/2}}}.
\end{align*}
\end{proof}

\begin{theorem}
For each $\ell\geq 2$ and $r\geq 0$, the limits \[\alpha_\ell:=\lim_{n\to \infty}\alpha_\ell(n)\, \quad \text{and} \quad \lambda_\ell(r)=\lim_{n\to \infty}\lambda_\ell(r,n)\]
 exist.  Furthermore, the convergence is exponential.
\end{theorem}
\begin{proof}
Follows directly from the definition of $\alpha_\ell(n)$ and Lemma \ref{limit}.
\end{proof}

\section{Recurrences}
From the estimates of the previous section, we know that
$\alpha_\ell(n)$ and $\lambda_\ell (r,n)$ both converge quickly to
$\alpha_\ell$ and $\lambda_\ell (r)$, respectively.  Thus, they can
be estimated to a few digits by explicit enumeration.

In order to evaluate $\avbl n \ell$ to dozens of decimal places, however,
we need a more efficient way
to calculate $\FH_\ell(\{p\},n)$. This can be done using the recurrence
formulas that we derive below. They are reformulations and generalizations of formulas given by Harborth
\cite{harborth} for sets of periods.

We first prove the following auxiliary claim.
\begin{lemma}\label{pq}
Let a word $w$ have a period $p<|w|$ and let $u$ be the prefix of $w$ of length $|w|-p$. Then $w$ has a period $q>p$ if and only if $u$ has a period $q-p$.
\end{lemma}
\begin{proof}
Note that $u$ is a border of $w$. The following conditions are easily seen to be equivalent:
\begin{itemize}
	\item $w$ has a period $q$,
	\item $w$ has a border of length $|w|-q$,
	\item $u$ has a border of length $|w|-q$,
	\item $u$ has a period $|u|-(|w|-q)$.
\end{itemize}
Since $|u|-(|w|-q)=(|w|-p)-(|w|-q)=q-p$, the proof is completed.
\end{proof}

\begin{theorem}\label{rec}
Let $P$ be a set of periods with $m=\min P$ and $\max P<n$.
Then
\begin{align}
\FH_\ell(P,n)=\GH_\ell(P,n)&-
\sum_{p=\ceil{m/2}}^{m-1} \HH_\ell(P,p,n)
\,,
\end{align}
where
\begin{align}\label{H}
\HH_\ell(P,p,n) :=
\begin{cases}
\FH_\ell\big((P-p)\cup\{p\},n-p\big),& \text{if $p<\ceil{n/2}$};\\[0.5em]
\ell^{2p-n}\cdot \FH_\ell(P-p,n-p), & \text{if $p\geq \ceil{n/2}$} \,.
\end{cases}
\end{align}
\end{theorem}
\begin{proof}
From $\GH_\ell(P,n)$ we have to subtract the number of words from $\Sigma_\ell^n$ that have periods $P$ but also a period smaller than $m$. We 
define, for each $1\leq p<m$, the set
\[
\+H_\ell(P,p,n)=\{w\in \Sigma_\ell^n \mid \text{$w$ has periods $P\cup \{p\}$, and no period $p'$ with $p<p'<m$}\}\,.
\]
If $p<\ceil{m/2}$ then $\+H_\ell(P,p,n)$ is empty, since a word $w\in \+H_\ell(P,p,n)$ also has a period $2p$, and $p<2p<m$ contradicts the definition of $\+H(p)$.
Moreover, the sets $\+H(P,p,n)$ are pairwise disjoint, and \[\GG_\ell(P,n)\setminus \FF_\ell(P,n)=\bigcup_{p=\ceil{m/2}}^{m-1}\+H_\ell(P,p,n)\,.\]
It remains to show that $H(p)$ is the cardinality of $\+H_\ell(P,p,n)$ for each $\ceil{m/2}\leq p<m-1$. 

Let $p<\ceil{n/2}$. We claim that $w\mapsto \pref[n-p] w$ is a one-to-one mapping of $\+H_\ell(P,p,n)$ to $\FF_\ell\big((P-p)\cup\{p\},n-p\big)$. Let $w\in \+H_\ell(P,p,n)$. By Lemma \ref{pq}, the word $\pref[n-p] w$ has periods $P-p$ and no period $p'-p$ with $p<p'<m$, that is, no period less than $m-p$. Since $m-p=\min\big((P-p)\cup\{p\}\big)$ and since $\pref[n-p] w$ also has a period $p$, we have $\pref[n-p] w\in \FF_\ell\big((P-p)\cup\{p\},n-p\big)$. Similarly, one can verify that if $v\in \FF_\ell\big((P-p)\cup\{p\},n-p\big)$, then $w_v:=(\pref[p] v)^{n/p}\in \+H_\ell(P,p,n)$ and $\pref[n-p]w_v=v$.

Let $p\geq \ceil{n/2}$. Again, using Lemma \ref{pq}, it is straightforward to verify that
\[\+H_\ell(P,p,n)=\{vuv \mid v\in \FF_\ell(P-p,n-p),u\in \Sigma_\ell^{2p-n}\}\,.\]

\end{proof}

If $\min P$ is small, then we can formulate a more explicit formula that uses the M\"obius $\mu$-function.

\begin{lemma}
Let $P$ be a set of periods with $m=\min P\leq \floor{n/2}+1$. Then
\begin{align}
\FH_\ell(P,n)=\sum_{d\divides m}\mu\left(\frac md\right) \GH_\ell\left(P\cup\left\{d\right\},n\right).
\end{align}
\end{lemma}
\begin{proof}
Let $w$ be a word of length $n$ with a period $m$ and let $p$ be the least period of $w$. Then, by the periodicity lemma,
we have that $p$ divides $m$, since $p<m$ implies $p+m-1\leq n$. Therefore, for each divisor $p$ of $m$,
\[\GH_\ell\left(P\cup\left\{p\right\},n\right)=\sum_{d\divides p}\FH_\ell(P\cup\left\{d \right\},n),\]
and the claim follows from M\"obius inversion.
\end{proof}

\section{Explicit formulas}
In this section 
we derive explicit formulas for
$\lambda_\ell(0)$ and $\lambda_\ell(1)$, which are the
asymptotic probabilities that a random word is unbordered,
or has longest border of length one, respectively.
These are two cases in which Lemma \ref{rec} yields a relatively
simple expression, since $\ceil{m/2}\geq \floor{n/2}$. 
\medskip
\subsection{Unbordered words}
The number of unbordered words satsifies
a well known recurrence formula (see, e.g., \cite[p.~143, Eq.~(34)]{harborth} for the binary case and \cite{Nielsen:1973} for the general case). The formula can be verified using Theorem \ref{rec} but we shall give an elementary proof. 
In this section,
let $u_n$ denote $\FH_\ell(\{n\},n)$, and let $t(n)$ denote
$\lambda_\ell(0,n)$.
\begin{theorem}\label{unb}
\[
u_n=
\begin{cases}
\ell, & \text{if $n=1$};\\
\ell(\ell-1) & \text{if $n=2$};\\
\ell\cdot u_{n-1}, & \text{if $n\geq 3$ is odd};\\
\ell\cdot u_{n-1}-u_{n/2}, & \text{if $n\geq 4$ is even}.
\end{cases}
\]
\end{theorem}
\begin{proof}
 For $k=1,2$, the verification is straightforward. Let $x$ and $y$ be nonempty words with $|x|=|y|$ and consider words $xy$, $xay$ and $xaby$ where $a$ and $b$ are letters. 

Since the shortest border of $xay$ has length at most $|x|$, the word $xy$ is unbordered if and only if $xay$ is. This proves $u_n=\ell\cdot u_{n-1}$ if $n$ is odd.

On the other hand, $xaby$ can have the shortest border of length $|x|+1$. Therefore, $xaby$ is unbordered  if and only if 
(i) $xy$ is unbordered and (ii) $xa\neq by$. Since the shortest border is itself unbordered, we obtain $u_n=\ell^2\cdot u_{n-2}-u_{n/2}=\ell\cdot u_{n-1}-u_{n/2}$ if $n$ is even.
\end{proof}
 Theorem \ref{unb} directly yields, for each $n\geq 1$,
\[t(2n+1)=t(2n)=t(2n-1)-t(n)\ell^{-n}\,.\]
Therefore 
\begin{align*}
t(2n)=t(1)+\sum_{i=2}^{2n} (t(i)-t(i-1))=1-\sum_{j=1}^{n}  t(j)\ell^{-j}\,.
\end{align*}
Defining the generating function $L_0(x)=\sum_{n\geq 1}t(n)x^n$, we get
\begin{align}
\lim_{n\to \infty}\lambda_\ell(0,n)=1-L_0\left(\frac 1\ell\right).
\end{align}
The next step is to obtain a functional equation for $L_0(x)$:
\begin{align*}
L_0(x)(1-x)&=t(1)x+\sum_{k\geq 2} (t(k)-t(k-1))x^k=\\
&=t(1)x+\sum_{j\geq 1}(t(2j)-t(2j-1))x^{2j}=\\
&=t(1)x-\sum_{j\geq 1}t(j)\ell^{-j}x^{2j}=x-L_0(x^2/\ell)\,.
\end{align*}
Therefore
\[L_0(x)=\frac x{1-x}-\frac{L_0(x^2/\ell)}{1-x}\,.\]
Successively substituting $x=1/\ell$, $x=1/\ell^3$, $x=1/\ell^7$, \dots, we get
\begin{align*}
L_0\left(\frac 1\ell\right)&=\frac 1{\ell-1}-\left(1+\frac 1{\ell-1}\right)L_0\left(\frac 1{\ell^3}\right),\\
L_0\left(\frac 1{\ell^3}\right)&=\frac 1{\ell^3-1}-\left(1+\frac 1{\ell^3-1}\right)L_0\left(\frac 1{\ell^7}\right),\\
&\vdots\\
L_0\left(\frac 1{\ell^{2^i-1}}\right)&=\frac 1{\ell^{2^i-1}-1}-\left(1+\frac 1{\ell^{2^i-1}-1}\right)L_0\left(\frac 1{\ell^{2^{i+1}-1}}\right).
\end{align*}
Since it is easy to see that
\[\lim_{n\to\infty} \frac1{\ell^{2^n-1}-1}
\prod_{i=1}^{n-1}
	\left(
	1+\frac 1{\ell^{2^i-1}-1}
	\right)
L_0\left(\frac 1{\ell^{2^n-1}}\right)=0\,,\]
we obtain

\begin{align*}
	L_0\left(\frac 1\ell\right)=\sum_{n\geq 1}
	\left(
	\frac{(-1)^{n+1}}{\ell^{2^n-1}-1}
	\prod_{i=1}^{n-1}
		{\left(
		1+\frac 1{\ell^{2^i-1}-1}
		\right)}
	\right).
	\end{align*}

A similar analysis was given previously by
\cite{Blom:1995}, although our analysis is slightly cleaner.

\subsection{Words with longest border of length 1}
  There is also a relatively simple recurrence for $\FH_\ell(\{n-1\},n)$, that is, for words with the longest border of length 1. The particular case $\ell=2$ was previously given by Harborth \cite[p.\ 143, Eq.~(36)]{harborth}.  In this section, we let $v_n$ denote $\FH_\ell(\{n-1\},n)$, and let $s(n)$ denote
  $\lambda_\ell(1,n)$.
\begin{theorem}\label{bor1}
\[
v_n=
\begin{cases}
0, & \text{if $n=1$};\\
\ell & \text{if $n=2$};\\
\ell\cdot v_{n-1}-v_{(n+1)/2}, & \text{if $n\geq 3$ is odd};\\
\ell\cdot v_{n-1}-(\ell-1)v_{n/2}, & \text{if $n\geq 4$ is even}.
\end{cases}
\]
\end{theorem}
\begin{proof}
 Verify that $v_1=0$ and $v_2=\ell$, and let $x$ and $y$ be nonempty words with $|x|=|y|$. Consider words $cxyc$, $cxayc$ and $cxabyc$ where $a,b,c$ are (not necessarily distinct) letters.

The letter $c$ is the longest border of the word $cxayc$ if and only if (i) $c$ is the longest border of $cxyc$ and (ii) $cxa\neq ayc$. Moreover, (i') $c$ is the shortest border of $cxyc$, and (ii') $cxa=ayc$ $(=cxc)$ if and only if $c$ is the shortest border of $cxc$. This implies $v_n=\ell\cdot v_{n-1}-v_{(n+1)/2}$ for $n\geq 3$ odd. 

Similarly, $c$ is the shortest border of $cxabyc$ if and only if (i) $c$ is the longest border of $cxyc$ and (ii) $cxa\neq byc$.
As above, we have to subtract the number of words $cxc$ with the longest border $c$. It follows that $v_n=\ell^2\cdot v_{n-2}-v_{n/2}=\ell v_{n-1}+(\ell-1)v_{n/2}$ for $n\geq 4$ even. 
\end{proof}
From Theorem~\ref{bor1}, we deduce
\begin{align}
s(2n)-s(2n-2)&=-s(n)\ell^{-n},& n&\geq 2,\label{bord1.1} \\ 
s(2n)-s(2n-1)&=(\ell-1)s(n)\ell^{-n}, & n&\geq 2,\\ 
s(2n+1)-s(2n)&=-s(n+1)\ell^{-n},&  n&\geq 1\,.
\end{align} 
Using \eqref{bord1.1}, we obtain
\begin{align*}
s(2n)=s(2)+\sum_{r=j}^{n} (s(2j)-s(2j-2))=1/\ell-\sum_{j=1}^{n}  s(j)\ell^{-j}\,.
\end{align*}
Defining the generating function $L_1(x)=\sum_{k\geq 1}s(k)x^k$, we then get
\[\lambda_\ell(1)=\frac 1\ell-L_1\left(\frac 1 \ell\right).\]
A functional equation for $L_1$ is obtained as follows:
\begin{align*}
L_1(x)&(1-x)=s(1)x +\sum_{k\geq 2}(s(k)-s(k-1))x^k=\\
					&=\frac 1\ell x^2  + \sum_{i\geq 1}(s(2i+1)-s(2i))x^{2i+1} + \sum_{i\geq 2}(s(2i)-s(2i-1))x^{2i}=\\
					&=\frac 1\ell x^2 - \sum_{i\geq 1}s(i+1)\ell^{-i}x^{2i+1}+	\sum_{i\geq 2}(\ell-1)s(i)\ell^{-i}x^{2i}=\\
					&=\frac 1\ell x^2 - \frac\ell x L_1\left(\frac{x^2}{\ell}\right) + (\ell-1)L_1\left(\frac{x^2}{\ell}\right).
\end{align*}
We have
\[L_1(x)=\frac{x^2}{\ell(1-x)}+\frac{\ell-1-\ell/x}{1-x}L_1\left(\frac{x^2}{\ell}\right),\]
and
\[L_1\left(\frac 1 {\ell^{2^i-1}}\right)=\frac {1}{\ell^{2^i}(\ell^{2^{i}-1}-1)}-\ell^{2^i-1}\frac{\ell^{2^i}-\ell+1}{\ell^{2^i-1}-1}L_1\left(\frac{1}{\ell^{2^{i+1}-1}}\right).\]
From here, we deduce
\[L_1\left(\frac 1 \ell\right)=\sum_{n\geq 1}\frac{(-1)^{n+1}}{\ell^{n+1}}\prod_{i=1}^j\frac{\ell^{2^{i-1}}-\ell+1}{\ell^{2^i-1}-1}\,.\]

We do not know how to obtain similar expressions for other border
lengths.

\section{Particular values}
Theorem \ref{rec}, as well as explicit formulas from the previous section allow fast computer evaluation of $\alpha_\ell(n)$ and $\lambda_\ell(r,n)$ for large $n$, and therefore also evaluation of $\lambda_\ell(r)$ and $\alpha_\ell$ with high precision. We list some rounded values in the following tables.   

\[
    \begin{tabular}{c|c}
      $\ell$ & $\alpha_\ell$\\
      \hline
      $2$ & $1.64116491178296695613$ \\
			$3$ & $0.68587617299708343978$ \\
			$4$ & $0.42195659003603599699$ \\
			$5$ & $0.30201601806282253073$ \\
			$10$& $0.12233344445364555354$ \\
			$50$& $0.02081648979722449000$
    \end{tabular}
\]

\medskip

\[
    \begin{tabular}{c|c}
      $r$ & $\lambda_2(r)$\\
      \hline
      $0$ & $0.26778684021788911238$ \\
			$1$ & $0.30042007151830329926$ \\
			$2$ & $0.19891874779036456415$ \\
			$3$ & $0.11216079483159432642$ \\
			$5$ & $0.03044609816129782975$ \\
			$10$& $0.00097577734413168807$
    \end{tabular}
\]

And some values of $\lambda\ell(r)$ rounded to four decimal digits:
\[
    \begin{tabular}{c|c|c|c|c}
      $\lambda_\ell(r)$ & $\ell=3$ & $\ell=4$ & $\ell=5$ &  $\ell=10$ \\
      \hline
      $r=0$ & $0.55698$ & $0.68775$ & $0.76006$ &  $0.89000$ \\
			$r=1$ & $0.28270$ & $0.23024$ & $0.19034$ &  $0.09890$ \\
			$r=2$ & $0.10547$ & $0.06126$ & $0.03961$ &  $0.00999$ \\
			$r=3$ & $0.03641$ & $0.01555$ & $0.00798$ &  $0.00100$ 
    \end{tabular}
\]

For example, we see that a long binary word chosen randomly has
about $27\%$ chance to be unbordered. 
A bit more probable, at $30\%$, is that such a word 
will have its longest border of length one.
Over a five-letter alphabet, more than three words out of four are unbordered,
on average.

Figure~\ref{fig1} shows the distribution of lengths of the shortest
period for binary words of length $n = 18$.

\begin{figure}[H]
\begin{center}
\includegraphics[width=4.5in]{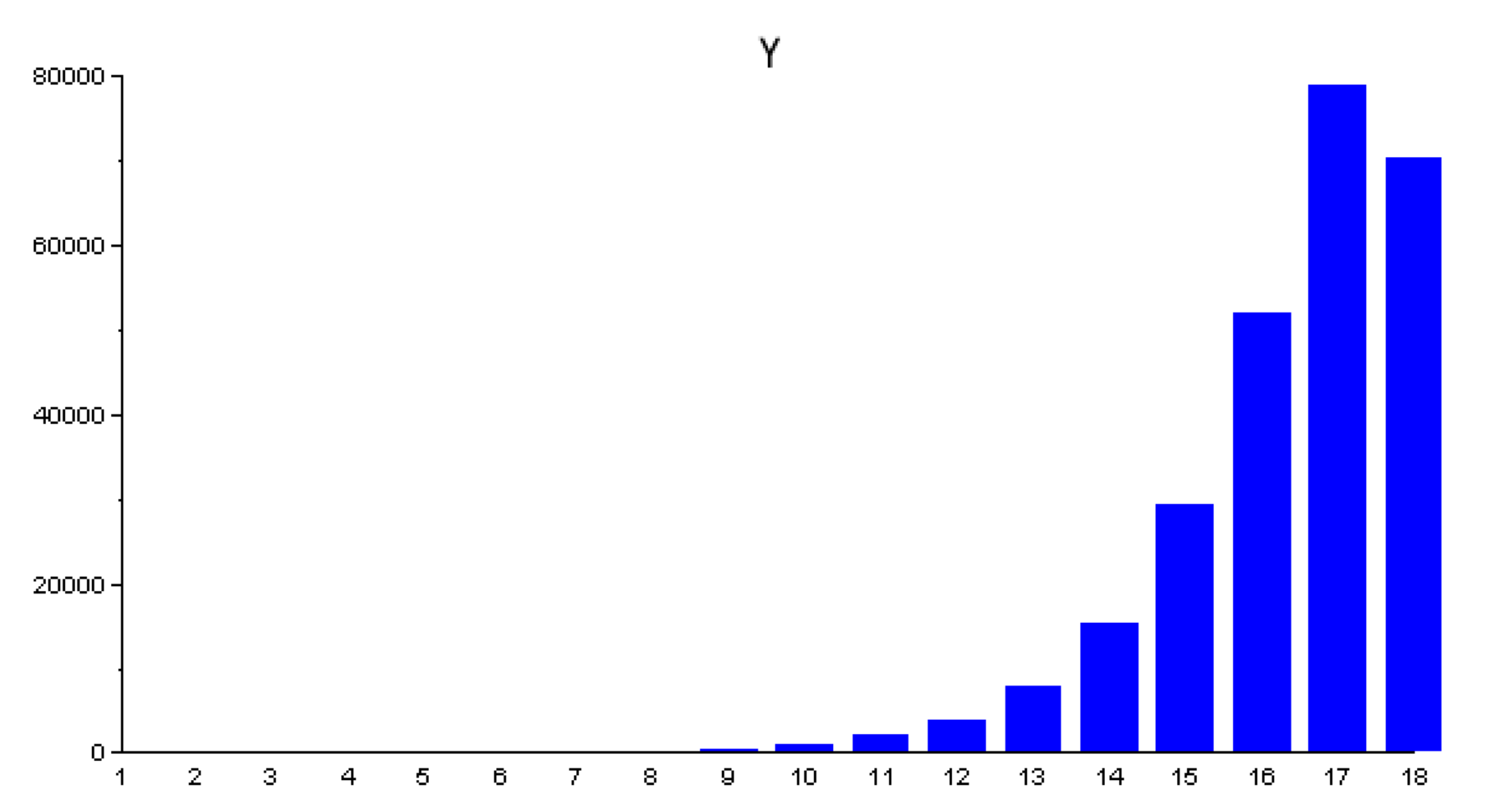}
\caption{Distribution of lengths of shortest period for binary words
of length $18$}
\label{squareorders}
\end{center}
\label{fig1}
\end{figure}

Our original motivation was a question about the average period of a binary word. The answer is, that the border of a binary word has asymptotically constant expected length, namely 
\[\alpha_2\doteq1.64116491178296695612774416940082554065953687825771543\dots\,.\] 

\section{Final remarks}

Recently there has been some interest in computing the expected value of the
largest unbordered factor of a word \cite{Loptev&Kucherov&Starikovskaya:2015}.
This is a related, but seemingly much harder, problem.

\newcommand{\noopsort}[1]{} \newcommand{\singleletter}[1]{#1}


 \end{document}